\documentclass{article}
\usepackage[utf8]{inputenc}
\usepackage{algorithm}
\usepackage[noend]{algorithmic}
\usepackage{graphicx}
\usepackage{amsmath}
\usepackage{amsthm}
\usepackage{amsfonts}
\usepackage{amssymb}

\newtheorem{theorem}{Theorem}
\newtheorem{corollary}{Corollary}[theorem]
\newtheorem{lemma}[theorem]{Lemma}
\newtheorem{conjecture}{Conjecture}

\theoremstyle{definition}

\theoremstyle{remark}

\title{Adaptive Fibonacci and Pairing Heaps}
\author{Andrew Frohmader}
\date{March 2020}

\begin{document}

\maketitle

\section{Introduction}
In this brief note, we present two adaptive heaps. See \cite{list-heaps} and \cite{Smooth_Heaps} for discussion of adaptive heaps. The first structure is a modification of the Fibonacci heap \cite{fibonacci:heap:1987}. The second, a relaxation of the Fibonacci like structure which is similar to a pairing heap \cite{pairing}. While the specification of the data structures is complete, the analysis is not. We welcome interested parties to complete it. 

\section{Adaptive Fibonacci Heap}
For ease of comparison, our structure mirrors that of CLRS Chapter 19 \cite{CLRS}. We include one additional pointer $root$ at the top level pointing to the first element added to the root list.

\subsection{Operations}
INSERT, FIND-MIN, UNION, and DECREASE-KEY are identical to CLRS 19 and have $O(1)$ amortized cost using the same potential function.

The entire difference is in how we build our trees during the EXTRACT-MIN operation. We add adaptivity to presorted input. We only need to modify the CONSOLIDATE$(H)$ operation of CLRS 19.2. For completeness, we include the EXTRACT-MIN code from 19.2. When we iterate over lists, we assume the iteration proceeds from the first element added to the list to the last element added to the list (oldest to newest).

\begin{algorithm}
\caption{EXTRACT-MIN$(H)$}
\begin{algorithmic}
    \STATE $z = H.min$
    \IF{$z \neq NIL$}
        \FOR{each child $x$ of $z$}
            \STATE add $x$ to the root list of $H$
            \STATE $x.p = NIL$
        \ENDFOR
        \STATE remove $z$ from the root list of $H$
        \IF{$z == z.right$}
            \STATE $H.min = NIL$
        \ELSE
            \STATE $H.min = z.right$
            \STATE CONSOLIDATE$(H)$
        \ENDIF
        $H.n = H.n-1$
    \ENDIF
    \RETURN $z$
\end{algorithmic}
\end{algorithm}

\begin{algorithm}
\caption{CONSOLIDATE$(H)$}
\begin{algorithmic}
    \STATE Initialize a new array $A[0 \dots D(H.n)]$ with all elements $NIL$
    \FOR{each node $x$ in the root list of $H$}
        \STATE APPEND$(x, x.degree, A)$
    \ENDFOR
    \STATE Create a new empty root list for $H$
    \FOR{each node $x$ in $A$}
        \IF{$x.parent == NIL$}
            \STATE insert $x$ into $H$'s root list
            \IF{$x < H.min$}
                \STATE $H.min = x$
            \ENDIF
        \ENDIF
    \ENDFOR
\end{algorithmic}
\end{algorithm}

\begin{algorithm}
\caption{APPEND$(x, d, A)$}
\begin{algorithmic}
    \STATE $y = A[d]$
    \IF{$y == NIL$}
        \STATE $break$
    \ELSIF{$y < x$}
        \IF{$y.degree == d$}
            \STATE Make $x$ a child of $y$ incrementing $y.degree$
        \ENDIF
        \IF{$y.parent == NIL$}
            \STATE APPEND$(y, y.degree, A)$
        \ENDIF
    \ELSIF{$y.parent == NIL$}
        \STATE Make $y$ a child of $x$ incrementing $x.degree$
    \ENDIF
    \STATE $A[d] = x$
\end{algorithmic}
\end{algorithm}

\newpage
\subsection{Example}
Say we insert the random sequence 
$$11, 13, 6, 10, 1, 8, 14, 12, 9, 5, 4, 3, 7, 2$$
into our heap and then call EXTRACT-MIN. EXTRACT-MIN removes $1$ and consolidates the remaining nodes. \textbf{Note:} the ``\textbf{else if} $y.parent == NIL$'' statement in APPEND causes $A[d] = x$ with $x.degree = d + 1$. These nodes are darkened in the example. The ``\textbf{if} $y.degree == d$'' statement ensures a darkened node is not given an additional child. 

\begin{figure}[ht]
    \centering
    \begin{tabular}{c c}
         (a) \includegraphics[width=5cm]{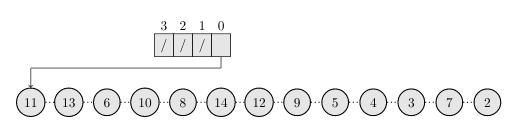} & (b) \includegraphics[width=5cm]{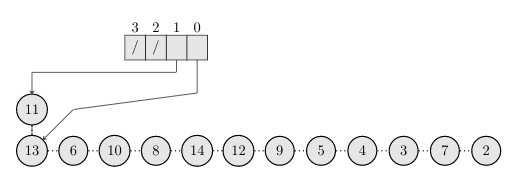}\\
         (c) \includegraphics[width=5cm]{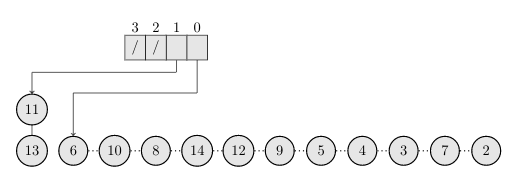} & (d) \includegraphics[width=5cm]{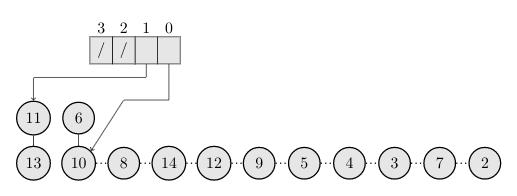}\\
         (e) \includegraphics[width=5cm]{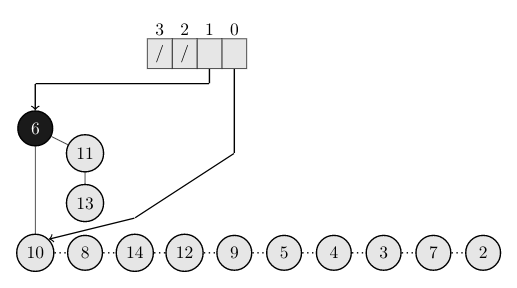} & (f) \includegraphics[width=5cm]{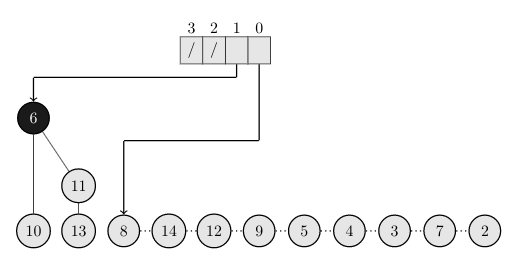}\\
        (g) \includegraphics[width=5cm]{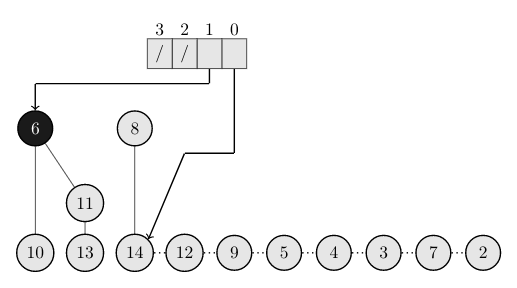} & (h) \includegraphics[width=5cm]{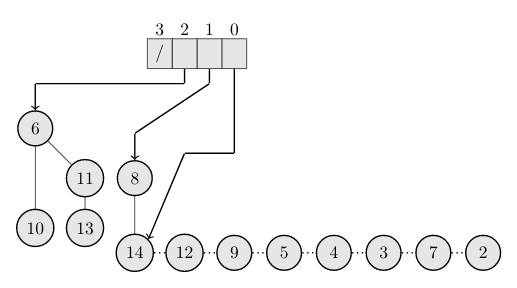}\\
    \end{tabular}
    \caption{Example of CONDSOLIDATE}
\end{figure}

\begin{figure}[ht]
    \centering
    \begin{tabular}{c c}
        (i) \includegraphics[width=5cm]{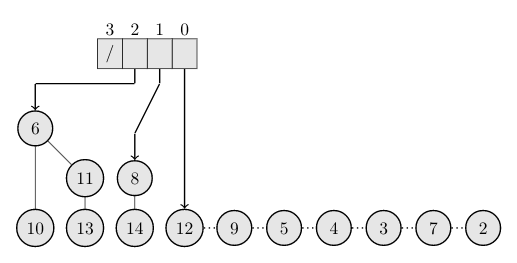} & (j) \includegraphics[height=4cm]{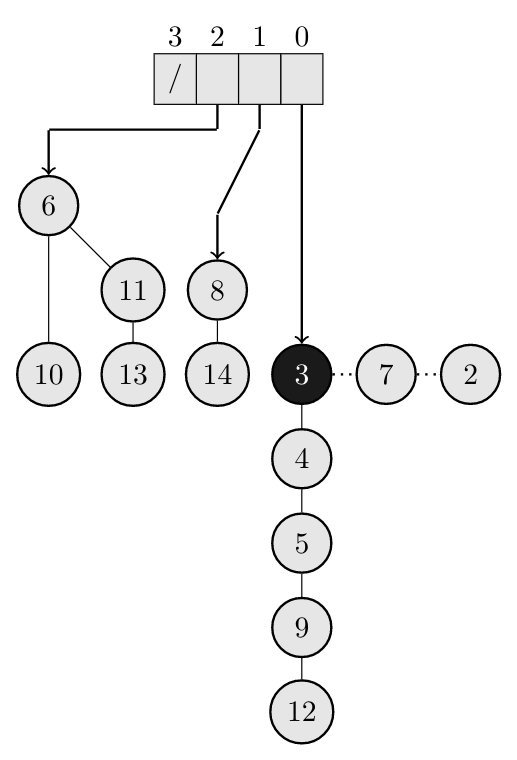}\\
        (k) \includegraphics[height=4cm]{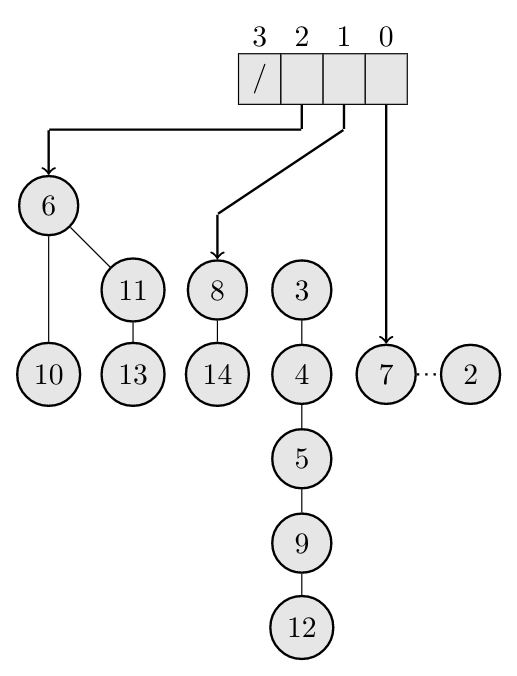} & (l) \includegraphics[height=4cm]{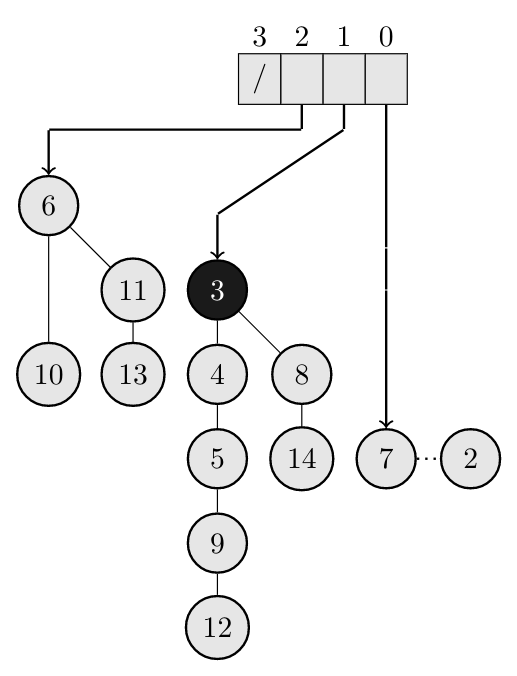}\\
        (m) \includegraphics[height=4cm]{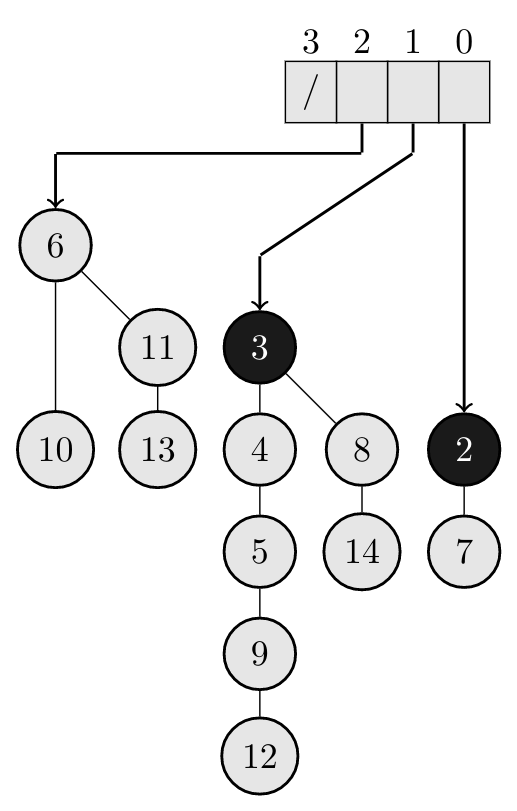} & (n) \includegraphics[height=4cm]{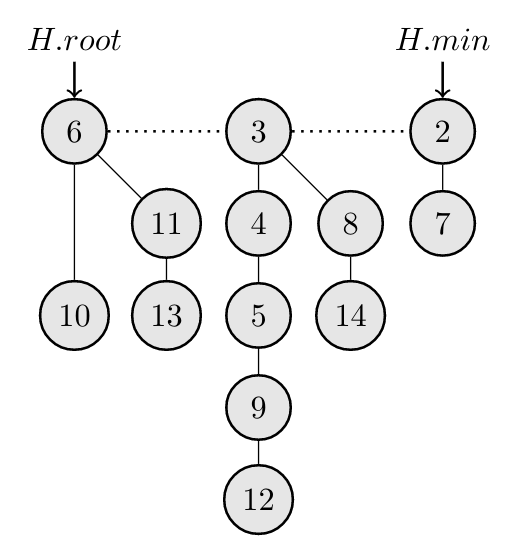}\\
    \end{tabular}
    
\caption{Example of CONDSOLIDATE continued.}
\end{figure}

\begin{lemma} \label{lemma:child_degree}
    Let $x$ be any node in the heap, and $k = x.degree$. Let $y_1, y_2, \dots , y_k$ denote the children of $x$ in the order in which they were linked to $x$, from earliest to latest. Then $y_1.degree \geq 0$ and $y_i.degree \geq i-2$ for $i = 2, 3, \dots , k$.
\end{lemma}

\begin{proof}
    Clearly $y_1.degree \geq 0$. For $i \geq 2$, we note that when $y_i$ was linked to $x$, all $y_1, y_2, \dots , y_{i-1}$ were children of $x$, and so we must have had $x.degree \geq i-1$. Let $x.degree = d$. We first show that when linked, $y_i.degree \geq d \geq i-1$. There are only two ways in which $y_i$ could have become a child of $x$ in CONSOLIDATE. 
    \begin{enumerate}
        \item If $x < y$, $A[d] = x$, and we call APPEND$(y, d, A)$. Then $y.degree = d$.
        \item If $x < y$, $A[d] = y$, $y.parent = NIL$, and we call APPEND$(x, d, A)$. Then $y.degree = d$ or $y.degree = d+1$. 
    \end{enumerate}
    Thus, $y.degree \geq d \geq i-1$. Since $y_i$ was linked, it has lost at most one child, since it would have been cut by CASCADING-CUT if it had lost two children. We conclude that $y_i.degree \geq i-2$. 
\end{proof}

With Lemma \ref{lemma:child_degree}, the remaining propositions can be proved nearly identically to CLRS. 

\begin{corollary}
    The maximum degree of all nodes is $O(\lg n)$.
\end{corollary}

\begin{theorem}
    EXTRACT-MIN runs in $O(\lg n)$ amortized time.
\end{theorem}
\begin{theorem}
    DECREASE-KEY runs in $O(1)$ amortized time.
\end{theorem}

We would like to have some dynamic optimality result like this:

\begin{conjecture}
    The above heap is competitive with all comparison based heaps that have an $O(1)$ amortized decrease key.
\end{conjecture}

\section{Pairing Like Heap}
\subsection{Structure}
 Use the same structure as used above except we do not need to store degree information in each node. 
 
 Two nice things about this structure. One, it makes the locally optimal choice of only leaving local min in the root list. Two, starting from a list of degree 0 nodes, the degree of any node in the final tree is $O(\lg n)$. Thus, it has better structural properties than the standard pairing heap. 
 
\subsection{Operations}

Again, operations are the same as for Fibonacci heap presented in CLRS 19 except for the CONSOLIDATE procedure in EXTRACT-MIN. The CONSOLIDATE code uses pointers $p$ for previous, $c$ for current and travels in loops around the circularly linked root list until only one root remains. A $cycle$ of CONSOLIDATE is one loop from beginning to end of root list.

\begin{algorithm}
\caption{CONSOLIDATE$(H)$}
\begin{algorithmic}
    \STATE $p = H.root$
    \STATE $c = p.right$
    \WHILE{$c \neq p$}
        \STATE $n = c.right$
        \IF{ $p < c$}
            \STATE Make $c$ a child of $p$ and remove $c$ from $root$ list
        \ELSIF{$p.parent = None$}
            \STATE Make $p$ a child of $c$ and remove $p$ from $root$ list
        \ENDIF
        \STATE $p = c$
        \STATE $c = n$
    \ENDWHILE
    \STATE $H.min = H.root = c$
\end{algorithmic}
\end{algorithm}

\begin{lemma}
   Let $l_i$ be the root list after the $i^{th}$ cycle of CONSOLIDATE. Then $l_{i+1}$ contains the local minimum of $l_i$.
\end{lemma}

\begin{lemma}
   Let $k$ be the number of roots in $H$'s root list prior to CONSOLIDATE. The increase in degree of any node caused by CONSOLIDATE is less than  $2\lg k$.
\end{lemma}

\begin{conjecture}
    The degree of any node $x$ in $H$ is $O(\lg n)$.
\end{conjecture}

\begin{conjecture} \label{conj:pairing-dk,em}
    DECREASE-KEY runs in $O(1)$ amortized time. EXTRACT-MIN runs in $O(\lg n)$ amortized time.
\end{conjecture}

\begin{conjecture}
    Let $X$ be a sequence of $n$ numbers, $m_1$ be the subsequence of local minimums in $X$ in the order in which they appear in $X$, and $m_i$ be the local minimums in $m_{i-1}$ in the order they appear in $m_{i-1}$. Say $|m_k| = 1$ and $|m_{k-1}| > 1$. Then a comparison sort is dynamically optimal if and only if it sorts every list in 
    $$ \Theta (nk)$$
\end{conjecture}

Again, something about optimality, although statement of this conjecture depends on Conjecture \ref{conj:pairing-dk,em}.
\begin{conjecture}
    The above heap is competitive with all comparison based heaps that have an $O(1)$ amortized decrease key.
\end{conjecture}

\subsection{Example}
Say we insert the same random sequence 
$$11, 13, 6, 10, 1, 8, 14, 12, 9, 5, 4, 3, 7, 2$$
into this heap and then call EXTRACT-MIN. EXTRACT-MIN removes $1$ and consolidates the remaining nodes. The resulting structure is presented below.

\begin{figure}
    \centering
    \includegraphics[height=5cm]{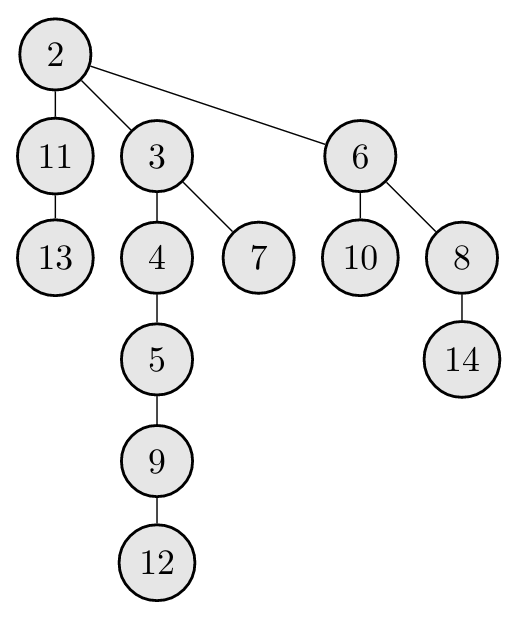}
    \caption{Example of pairing-like structure}
    \label{fig:my_label}
\end{figure}

\bibliographystyle{plain}
\bibliography{references}

\begin{thebibliography}{1}

\bibitem{CLRS}
Thomas~H. Cormen, Charles~E. Leiserson, Ronald~L. Rivest, and Clifford Stein.
\newblock {\em Introduction to Algorithms, Third Edition}.
\newblock The MIT Press, 3rd edition, 2009.

\bibitem{pairing}
Michael Fredman, Robert Sedgewick, Daniel Sleator, and Robert Tarjan.
\newblock The pairing heap: A new form of self-adjusting heap.
\newblock {\em Algorithmica}, 1:111--129, 11 1986.

\bibitem{fibonacci:heap:1987}
Michael~L. Fredman and Robert~Endre Tarjan.
\newblock Fibonacci heaps and their uses in improved network optimization
  algorithms.
\newblock {\em J. ACM}, 34(3):596--615, July 1987.

\bibitem{list-heaps}
Andrew Frohmader.
\newblock List heaps.
\newblock {\em CoRR}, abs/1802.05662, 2018.

\bibitem{Smooth_Heaps}
L\'{a}szl\'{o} Kozma and Thatchaphol Saranurak.
\newblock Smooth heaps and a dual view of self-adjusting data structures.
\newblock In {\em Proceedings of the 50th Annual ACM SIGACT Symposium on Theory
  of Computing}, STOC 2018, page 801–814, New York, NY, USA, 2018.
  Association for Computing Machinery.

\end{thebibliography}

\end{document}